\documentclass[superscriptaddress,aps,prl,10pt,twocolumn,footinbib,floatfix,nobalancelastpage]{revtex4-1}

\usepackage{microtype}
\usepackage{amsmath}
\usepackage{amssymb}
\usepackage{quantum}
\usepackage[final]{graphicx}
\usepackage{color}
\usepackage[usenames,dvipsnames]{xcolor}
\usepackage[colorlinks,linktocpage,hypertexnames=false,pdfpagelabels]{hyperref}
\usepackage[amsmath,thmmarks,hyperref]{ntheorem}
\usepackage[capitalise,poorman]{cleveref}

\newtheorem{theorem}{Theorem}
\newtheorem{lemma}[theorem]{Lemma}

\theoremstyle{nonumberplain}
\theorembodyfont{\normalfont}
\theoremsymbol{$\Box$}
\newtheorem{proof}{Proof}

\crefname{algorithm}{Algorithm}{Algorithm}
\crefname{part}{Part}{Parts}
\crefname{step}{Step}{Steps}

\newcommand{\For}[1]{\textbf{For} #1}
\newcommand{\While}[1]{\textbf{While} #1}

\newcommand{\Input}[1]{\textbf{Input:} #1\\}
\newcommand{\Output}[1]{\textbf{Output:} #1\\}

\newcommand{\keyword}[1]{\textit{#1}}

\renewcommand{\paragraph}[1]{\textbf{\textit{#1}}}

\makeatletter
\newcommand{\TODO}[2][\@empty]{%
  \begingroup%
    \def\@tmpa{#2}%
    \ifx\@tmpa\@empty%
      \color{red}\textup{\textrm{[TODO]}}%
    \else%
      \ifx #1\@empty%
        \color{red}\textup{\textrm{[TODO: #2]}}%
      \else%
        \color{red}\textup{\textrm{[TODO~(#1): #2]}}%
      \fi%
    \fi%
  \endgroup}
\newcommand*{\balancecolsandclearpage}{%
  \close@column@grid
  \clearpage
  \twocolumngrid}
\makeatother

\begin{document}

\title{Preparing topological PEPS on a quantum computer}

\author{Martin Schwarz}
\affiliation{Vienna Center for Quantum Science and Technology,
  Faculty of Physics, University of Vienna, Vienna, Austria}
\author{Toby S.\ Cubitt}
\affiliation{Departamento de An\'alisis Matem\'atico, Universidad Complutense
  de Madrid, Plaza de Ciencias 3, Ciudad Universitaria, 28040 Madrid, Spain}
\author{Kristan Temme}
\affiliation{Center for Theoretical Physics,
  Massachusetts Institute of Technology,
  77 Massachusetts Avenue, Cambridge MA 02139-4307, USA}
\author{Frank Verstraete}
\affiliation{Vienna Center for Quantum Science and Technology,
  Faculty of Physics, University of Vienna, Vienna, Austria}
\author{David Perez-Garcia}
\affiliation{Departamento de An\'alisis Matem\'atico, Universidad Complutense
  de Madrid, Plaza de Ciencias 3, Ciudad Universitaria, 28040 Madrid, Spain}

\begin{abstract}
  Simulating of exotic phases of matter that are not amenable to classical
  techniques is one of the most important potential applications of quantum
  information processing. We present an efficient algorithm for preparing a
  large class of topological quantum states -- the G\nobreakdash-injective
  Projected Entangled Pair States (PEPS) -- on a quantum computer. Important
  examples include the resonant valence bond (RVB) states, conjectured to be
  topological spin liquids. The runtime of the algorithm scales polynomially
  with the condition number of the PEPS projectors, and inverse-polynomially
  in the spectral gap of the PEPS parent Hamiltonian.
\end{abstract}

\maketitle

Creating and studying exotic phases of matter is one of the most challenging
goals in contemporary physics. The increasingly sophisticated simulation
abilities of systems such as cold atoms in optical lattices, trapped ions or
superconducting qubits make this challenge accessible by means of Feynmann's
original idea of using highly controllable quantum systems in order to
simulate other quantum systems. Among those exotic phases, non-abelian
topologically ordered states and topological spin liquids -- such as
resonating valence bond (RVB) states in frustrated lattices -- are probably
the holy grails of this area of quantum state engineering. Progress on the
creation of such exotic phases in various experimental systems has accelerated
rapidly in recent years, including cold atoms~\cite{coldatoms}, ion
traps~\cite{iontraps}, photonic devices~\cite{photonic} and superconducting
devices~\cite{superconducting}.

Recently~\cite{Martin}, a very general way of constructing quantum states on a
quantum computer was proposed. The wide applicability of the method lies in
the fact that there is a variational class of quantum states, called
Projective Entangled Pair States (PEPS), which has a simple local description
but is nonetheless complex enough to approximate the low-energy sector of
local Hamiltonians. (A review of the analytical and numerical evidence for
this can be found in~\cite{SCPG10} and the references therein.) However, a
crucial technical assumption in the main result of~\cite{Martin}, called
`injectivity', excludes any possibility of constructing quantum states with
topological order.

The main aim of this article is to significantly extend the result of
Ref.~\cite{Martin} to include exotic topological quantum phases, by
proving:\\
\noindent\paragraph{Main result} For any finite group $G$, a
G\nobreakdash-injective PEPS can be prepared on a quantum computer in
polynomial time, 
when the inverse ground state gap of the associated parent Hamiltonian scales
at most polynomially in the system size.

`G\nobreakdash-injectivity', introduced only recently in~\cite{SCPG10} (and
explained more fully below), is a substantially weaker requirement than
injectivity, which explicitly allows for topological order. A compelling
example of the significance of this result is the very recently proven fact
that the resonating valence bond (RVB) state in the Kagome lattice
(conjectured to be a topological spin liquid), is a $\mathbb{Z}_2$-injective
PEPS, with numerical evidence that the gap assumption is also
verified~\cite{PSPGC12}. Our result therefore gives one way in which the RVB
state (and other topological states) can be prepared efficiently on a general
quantum simulator. Engineering exotic quantum states by quantum simulation
complements research aimed at finding materials that directly exhibit
topological behaviour, and is already beginning to bear fruit
experimentally~\cite{coldatoms,iontraps,photonic,superconducting}.

In the following section, we summarise basic notions of PEPS required in this
work, and introduce the class of G\nobreakdash-injective PEPS which includes
many of the important topological quantum states. We then briefly review the
algorithm of Ref.~\cite{Martin} for preparing injective (non-topological)
PEPS, before proceeding to show how this algorithm can be extended to the much
larger class of G\nobreakdash-injective PEPS, thereby allowing efficient
preparation of many exotic topological quantum states. Finally, we close with
some concluding remarks and open questions.

\paragraph{Projected Entangled Pair States}
For simplicity, we will focus in this letter on PEPS defined on a square
lattice, but the results can be generalized to other lattices. An
(unnormalized) PEPS can be described as follows. Place maximally entangled
states of dimension $D$ along all edges of the lattice. To each vertex $\nu$,
apply a linear map $A^{\nu}: (\CC^D)^{\otimes 4} \rightarrow \CC^d$ in the
four $D$-dimensional systems labeled by $l,t,r,b$ (for `left', `top', `right'
and `bottom'), where $A^\nu=\sum_{i;l,t,r,b} A^\nu_{i;ltrb}\ketbra{i}{ltrb}$.
The resulting vector in $(\CC^d)^{\otimes N}$ is the unnormalized PEPS, $N$
being the number of vertices in the lattice. For the purposes of this work,
since local unitaries do not change the complexity of preparing a state, by
taking the polar decomposition of $A$ we can assume without loss of generality
that $A$ is positive-semidefinite. When $A$ is invertible, we call the PEPS
\keyword{injective}~\cite{SCPG10}.

A particularly interesting class of PEPS is the class of G\nobreakdash-isometric PEPS,
defined for any finite group $G$ as follows. Take a semi-regular
representation of $G$ \cite{SCPG10} -- that is, a representation $U_g=
\oplus_\alpha V_g^{\alpha} \otimes \mathbbm{1}_{r_\alpha}$ having at least one
copy of each irrep $\alpha$. Note that the regular representation is exactly
the one for which $r_\alpha$ is the dimension $d_\alpha$ of the irrep
$V_g^{\alpha}$ for all $\alpha$. We can define the re-weighting map
\begin{equation}\label{eq:delta}
  \Delta = \oplus_\alpha
    \left(\frac{d_\alpha}{r_\alpha}\right)^{\frac{1}{4}}\mathbbm{1}_{d_\alpha}
    \otimes \mathbbm{1}_{r_\alpha}
\end{equation}
which is real, diagonal, commutes with $U_g$ and satisfies
$\tr{\Delta^4U_g}=\abs{G}\delta_{g,e}$. (For the regular representation
$\Delta = \mathbbm{1}$.) The PEPS is then defined by taking, for all $\nu$:
\begin{equation}\label{eq:proj}
  A^{\nu}
  = \frac{1}{\abs{G}}\sum_{g\in G}\Delta\bar{U_g}
    \ox \Delta\bar{U_g} \ox \Delta U_g\otimes  \Delta U_g\,.
\end{equation}
G\nobreakdash-isometric PEPS were originally defined in~\cite{SCPG10} only for
the regular representation, and shown in that case to be exactly the quantum
double models of Kitaev~\cite{kitaev2003}. Here, we generalise the definition
of G\nobreakdash-isometric to any semi-regular representation~\footnote{The
  generalisation of `G-isometric to any semi-regular representation is
  justified by the fact that, for a given $G$, all G-isometric PEPS are
  equivalent (up to isometries) to the G-isometric PEPS for the regular
  representation by decomposing the tensor in a different way. A proof is
  given in the Appendix.}. If, on top of a G\nobreakdash-isometric PEPS, we
apply a further invertible (and w.l.o.g.\ positive-definite) linear map
$A^{\nu}:\CC^d\rightarrow \CC^d$, we obtain a `G\nobreakdash-injective'
PEPS~\cite{SCPG10}. (Here, $d$ is the dimension of the symmetric subspace
associated with the group.) The parallel with plain injective PEPS is clear.
Both are defined by invertible maps on top of a G\nobreakdash-isometric PEPS.
In the case of injective PEPS, the group is the trivial one and the
representation is simply $\id_d$ ($d$ copies of the left-regular
representation of the trivial group).

G\nobreakdash-isometric PEPS have very nice properties, coming from their
topological character, which are inherited by the more general
G\nobreakdash-injective PEPS. For instance, for each G\nobreakdash-isometric
PEPS $\ket{\psi}$ there exists a local frustration-free Hamiltonian (called
the PEPS ``parent Hamiltonian''~\cite{SCPG10}), consisting of commuting
projectors and having as ground space the subspace (over-)spanned by
$\{\ket{\psi;K}: K=(g,h), [g,h]=0\}$. (Here, $\ket{\psi;K}$ is the PEPS
obtained by the same maps $A$, except that we first apply an additional
$U_g^{\otimes V}$ to exactly one vertical strip $V$ and $U_h^{\otimes H}$ to
exactly one horizontal strip $H$ in the initial collection of maximally
entangled states~\cite{SCPG10}). This generalises to G\nobreakdash-injective
PEPS, except that the local Hamiltonian terms are no longer necessarily
commuting projectors.

We will denote by $\ket{A^1\cdots A^t}$ the G\nobreakdash-injective PEPS
defined by applying the map $A^j$ to vertex $j$ for $j=1,\dots,t$ (and
identity to the rest of the vertices) on top of the G\nobreakdash-isometric
PEPS, and define the states $\ket{A^1\cdots A^t; K}$ analogously to above,
which again (over-)span the ground space of a frustration-free local parent
Hamiltonian $H_t$.

\paragraph{Preparing injective PEPS}
We first briefly review the algorithm of~\cite{Martin} for preparing injective
PEPS on a quantum computer. Let $H_t$ be the parent Hamiltonian of the
partially constructed state $\ket{A^1\cdots A^t}$. The algorithm starts at
$t=0$ with maximally entangled states between all pairs of adjacent sites in
the lattice, and proceeds by successively projecting onto the ground states of
$H_t$ for $t=1\dots N$ until the final state $\ket{A^1\cdots A^N}$ is reached.

Since the ground state $P_t$ of $H_t$ is a complex, many-body quantum state,
it is not immediately clear (i)~how to efficiently perform the projective
measurement $\{P_t,P_t^\perp\}$ onto the ground state. Furthermore,
measurement in quantum mechanics is probabilistic, so even if this measurement
can be performed, it is not at all clear (ii)~how to guarantee the desired
outcome $P_t$.

The answer to~(i) is to run the coherent quantum phase estimation
algorithm~\cite{NC00,KOS07} for the unitary generated by time-evolution
under $H_t$. (Time-evolution under the local Hamiltonian $H_t$ can be
simulated efficiently by standard Hamiltonian simulation
techniques~\cite{BACS07}.) If $\sum_k\alpha_k\ket{\psi_k}$ is the
initial state expanded in the eigenbasis of $H_t$, then the phase estimation
entangles this register with an output register containing an estimate of the
corresponding eigenvalue: $\sum_k\alpha_k\ket{\psi_k}\ket{E_k}$. Performing a
partial measurement on the output register to determine if its value is less
than $\Delta_t$ (the spectral gap of $H_t$) completes the implementation of
the measurement $\{P_t,P_t^\perp\}$. (See~\cite{Martin} for full details.)

The solution to~(ii) is more subtle, and makes use of Camille Jordan's lemma
of 1875 on the simultaneous block diagonalization of two projectors, which we
first recall (see also the related CS~decomposition~\cite{Golub+vanLoan}):
\begin{lemma}[Jordan \cite{jordan1875}]\label{Jordan}
  Let $R$ and $Q$ be two projectors with rank $s_r = \rank R$ and $s_q = \rank
  Q$ respectively. Then both projectors can be decomposed simultaneously in
  the form
  \begin{equation}
    R = \bigoplus_{k=1}^{s_r} R_k \quad Q = \bigoplus_{k=1}^{s_q} Q_k,
  \end{equation}
  where $R_k,Q_k$ denote rank-1 projectors acting on one- or two-dimensional
  subspaces. The eigenvectors
  $\ket{r_k},\ket{r_k^\perp}$ and $\ket{q_k},\ket{q_k^\perp}$ of the $2 \times
  2$ projectors $R_k$ and $Q_k$ are related by
  \begin{align*}
    \ket{r_k}\;
      &= \sqrt{d_k}\ket{q_k} + \sqrt{1-d_k} \ket{q_k^\perp}\\
    \ket{r_k^{\perp}}
      &= -\sqrt{1-d_k} \ket{q_k} + \sqrt{d_k}\ket{q_k^\perp}\\
    \ket{q_k}\;
      &= \sqrt{d_k}\ket{r_k} - \sqrt{1-d_k} \ket{r_k^\perp}\\
    \ket{q_k^{\perp}}
      &= \sqrt{1-d_k}\ket{r_k} + \sqrt{d_k} \ket{r_k^\perp}.
  \end{align*}
\end{lemma}

Ref.~\cite{Martin} shows that if the current state is in the block
containing the ground state of $H_t$, then the PEPS structure guarantees that
the probability of a successful projection onto $P_{t+1}$ is lower-bounded by
$\kappa(A^{t+1})^{-2}$, where $\kappa(A^{t+1})$ is the condition number of the
matrix $A^{t+1}$.
Assume for induction that we have already successfully prepared the (unique)
ground state of $H_t$. We first attempt to project from this state onto the
unique ground state of $H_{t+1}$ by measuring $\{P_{t+1},P_{t+1}^\perp\}$. If
this fails, we attempt to project back to the state we started from by
measuring $\{P_t,P_t^\perp\}$, a technique introduced by Marriott and
Watrous \cite{MW05} in the context of QMA-amplification. If this ``rewind''
measurement succeeds, then
we're back to where we started and can try again. What if the ``rewind''
measurement fails? By \cref{Jordan}, we can only be in the excited
state from the same block, so we can still try to project ``forwards''
with the same lower bound on the success probability. Thus iterating forwards
and backwards measurements until success generates a Markov chain with
successful projection onto the ground state of $H_{t+1}$ as the unique
absorbing state. Moreover, since the success probability in each step is
bounded away from zero, this converges rapidly to the desired state, allowing
us to move from the ground state of $H_t$ to the ground state of $H_{t+1}$ in
polynomial time.

\paragraph{Preparing G-injective PEPS}
Consider the algorithm of the preceding section from the perspective of
G\nobreakdash-injective PEPS. An injective PEPS can always be viewed as a
G\nobreakdash-injective PEPS for the representation $\id$ of the trivial
group. The algorithm starts from the state consisting of maximally-entangled
pairs between each site, and transforms this into the desired state by
projecting onto the ground states of a sequence of injective parent
Hamiltonians. But the initial state is none other than the
G\nobreakdash-\emph{isometric} PEPS corresponding to the representation $\id$
of the trivial group. This hints at a generalisation of the algorithm to
G\nobreakdash-injective PEPS for arbitrary groups~G: start by preparing the
corresponding G\nobreakdash-\emph{isometric} PEPS, and successively transform
this into the desired G\nobreakdash-injective PEPS by projecting onto the
ground states of the sequence of G\nobreakdash-injective parent Hamiltonians
(see \cref{alg:G-injective}).

\begin{table}[!hbtp]
  \raggedright\vspace{.5em}
  \Input{G\nobreakdash-injective $A^v$ defined on an $N$-vertex lattice;
    $\epsilon>0$.}
  \Output{$\ket{\psi}\in\vspan{\ket{A^1,\dots,A^t;K}}$ with
    probability $\geq 1-\epsilon$.}
  \begin{enumerate}
    \item Prepare corresponding G\nobreakdash-isometric PEPS
      \label[step]{step:G-isometric}
    \item \For{$t=1$ to $N$} \label[step]{step:vertexloop}
      \begin{enumerate}
        \item Measure $\{P_{t+1},P_{t+1}^\perp\}$ on $\ket{\psi}$.
          \label[step]{step:measurenext}
        \item \While{measurement outcome is $P_{t+1}^\perp$}
          \label[step]{step:failedloop}
          \begin{enumerate}
            \item Measure $\{P_t,P_t^\perp\}$.
              \label[step]{step:measureundo}
            \item Measure $\{P_{t+1},P_{t+1}^\perp\}$.
              \label[step]{step:measureredo}
          \end{enumerate}
      \end{enumerate}
  \end{enumerate}
  \caption{Preparing a G-injective PEPS. $H_t$ ($t=1 \dots N$) is the parent
    Hamiltonian for the G-injective PEPS $\ket{A^1\dots A^t}$, $P_t$ the
    projector onto its ground state subspace. Note that by specifying
    $A^v$, we are implicitly selecting a particular semi-regular
    representation of $G$.}
  \label{alg:G-injective}
\end{table}

There are, however, two obstacles to implementing this approach. (i)~The
initial G\nobreakdash-isometric PEPS can be a substantially more complicated
many-body quantum state than the trivial product of maximally-entangled pairs
we must prepare in the injective case. (ii)~Since G\nobreakdash-injective
parent Hamiltonians are topological, they have degenerate ground state
subspaces. But the Marriott-Watrous ``rewinding trick'' \cite{MW05} relies on
the measurement projectors being 1-dimensional; it breaks down in general for
higher-dimensional projectors.

It turns out that there is a direct solution to~(i). Ref.~\cite{SCPG10} proves
that, for any group $G$, the parent Hamiltonian of the G\nobreakdash-isometric
PEPS for the regular representation corresponds precisely to a quantum-double
model~\cite{kitaev2003,Schuch}. But Ref.~\cite{Aguado} shows that ground
states of quantum-double models can be generated exactly by a polynomial-size
quantum circuit. We can therefore use this circuit to efficiently prepare the
G\nobreakdash-isometric PEPS for the regular representation of $G$.
What if the representation we require is semi-regular? In fact the
G\nobreakdash-isometric PEPS for any semi-regular representation is equivalent
to the one for the regular representation (up to local isometries), by simply
regrouping the tensors. (See the appendix for a proof.)

The second obstacle is more delicate. As described above, the Marriott-Watrous
``rewinding trick'' used in the injective case~\cite{Martin} works because,
thanks to the Hamiltonians $H_t$ and $H_{t+1}$ at each step having unique
ground states, there is only ever one $2\times 2$ block involved in the
back-and-forth measurements.
%
However, in the G\nobreakdash-injective case, the Hamiltonians no longer have
unique ground states, and there are multiple $2\times 2$ blocks corresponding
to different ground states. Thus when we ``rewind'' a failed measurement, the
backwards measurement could project us back into any superposition of states
from any of the blocks corresponding to the ground state subspace. Now,
$A^{t+1}$ is only invertible on the $G$-symmetric subspace, so it necessarily
has some zero eigenvalues. Hence $\kappa(A^{t+1})=\infty$ and the lower bound
$\kappa(A^{t+1})^{-2}=0$ on the probability of a successful forward
measurement is useless. Although there may still exist some ground state
$\ket{\psi_t^1}$ of $H_t$ which has positive probability of successful forward
transition to a ground state of $H_{t+1}$, this does not rule out existence of
another ground state $\ket{\psi_t^2}$ of $H_t$ for which the probability of a
successful forward transition is~0.
In the worst case, if a forward measurement fails and we end up in a state
$\ket{\varphi_{t+1}^\perp}$, the rewinding step could have probability~1 of
transitioning back to $\ket{\psi_t^2}$, so that we remain stuck forever
bouncing back and forth between $\ket{\psi_t^2}$ and
$\ket{\varphi_{t+1}^\perp}$.

To overcome this, we must show that \emph{if we start from the
  G\nobreakdash-isometric state}, then the structure of
G\nobreakdash-injective PEPS ensures that this situation can never occur. To
prove this, we need the following key lemma:
\begin{lemma}\label{dmin}
  Let $P_t$ and $P_{t+1}$ denote two projectors on the ground state subspace
  of the partial PEPS parent Hamiltonains $H_t$ and $H_{t+1}$ for
  $\ket{A^1\dots A^t}$ and $\ket{A^1\dots A^t,A^{t+1}}$. The overlap $d_k$
  between $P_t$ and $P_{t+1}$ (cf.\ \cref{Jordan}) is lower-bounded by
  $d_{\min} \geq \kappa(A^{t+1}\vert_{S_G})^{-2}$, where
  $\kappa(A^{t+1}\vert_{S_G}):=\sigma_{\max}(A_{t+1}|_{S_G})/\sigma_{\min}(A_{t+1}|_{S_G})$
  is the condition number restricted to the $G$-symmetric subspace~$S_G$.
\end{lemma}
\begin{proof}
  The minimum overlap $d_{\min}$ between projectors $P_t$ and $P_{t+1}$ is
  given by
  \begin{equation}\label{eq:dmin}
    d_{\min}
      = \min_{\ket{\psi_t}} \max_{\ket{\psi_{t+1}}}
        \Abs{\braket{\psi_t}{\psi_{t+1}}}^2,
  \end{equation}
  where $\ket{\psi_t}$ and $\ket{\psi_{t+1}}$ are states in the respective
  ground state subspaces $\ker H_t$ and $\ker H_{t+1}$.

  Now, $\ker H_t$ is spanned by the partially constructed PEPS
  $\ket{A^1,\dots,A^t;K}$, with different boundary conditions $K$ giving
  different ground states). Thus we can decompose any $\ket{\psi_t}\in\ker
  H_t$ as a linear combination of partial PEPS: $\ket{\psi_t} = \sum c_k
  \ket{A^1,\dots,A^t;K^k}$.

  $H_{t+1}$ is obtained from $H_t$ by replacing all the G\nobreakdash-isometric local
  Hamiltonian terms at one vertex with the G\nobreakdash-injective terms. So applying
  $A_{t+1}$ to \emph{any} $\ket{A^1,\dots,A^t;K}$ takes us to the next ground
  state subspace. Therefore, the state
  \begin{equation}
    \ket{\varphi_{t+1}} =
    \frac{A_{t+1}\ket{\psi_t}}
         {\sqrt{\bra{\psi_t}A_{t+1}^\dagger A_{t+1}\ket{\psi_t}}}
  \end{equation}
  is contained in $\ker H_{t+1}$. Choosing $\ket{\psi_{t+1}} =
  \ket{\varphi_{t+1}}$ in \cref{eq:dmin}, we obtain the lower bound
  \begin{equation}
    d_{\min}
    \geq \min_{\ket{\psi_t}} \Abs{\braket{\psi_t}{\varphi_{t+1}}}^2
    \geq \min_{\ket{\psi_t}} \frac{\Abs{\bra{\psi_t}A_{t+1}\ket{\psi_t}}^2}
                               {\bra{\psi_t}A_{t+1}^\dagger A_{t+1}\ket{\psi_t}}.
  \end{equation}

  It is immediate from the definition of G\nobreakdash-injective PEPS that the ground
  states of $H_t$ are symmetric, so that the projector $P_t$ is supported on
  the symmetric subspace $S_G$. Thus the minimisation is over symmetric states
  and, recalling that w.l.o.g.\ $A_t$ is positive-semidefinite, we obtain the
  claimed bound
  \begin{equation}
    d_{\min}
    \geq \min_{\ket{\psi_t}}
      \frac{\bra{\psi_t}\left. A_{t+1} \right|_{S_G}\ket{\psi_t}^2}
           {\bra{\psi_t}\left . A^2_{t+1} \right|_{S_G}\ket{\psi_t}}
    \geq \frac{\sigma_{\min}\left(\left . A_{t+1}\right|_{S_G}\right)^2}
              {\sigma_{\max}\left(\left . A_{t+1} \right|_{S_G}\right)^2},
  \end{equation}
  by the variational characterisation of eigenvalues.
\end{proof}

\paragraph{Runtime}
We are now in a position to establish the runtime of the algorithm given in
\cref{alg:G-injective}. We start by bounding the failure probability of
growing the partial PEPS by a single site.

\begin{figure}[ht]
  \begin{center}
    \includegraphics[scale=1]{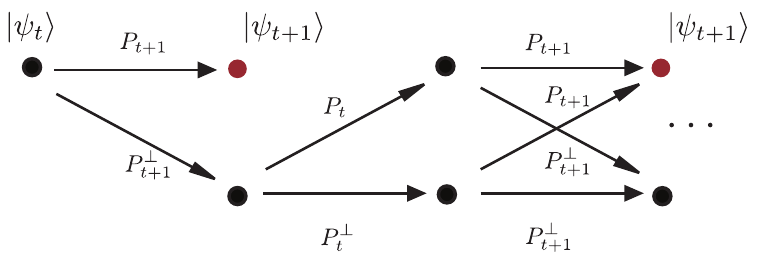}
    \caption{The sequence of outcomes of the binary measurements
      $\{P_t,P_t^\perp\}$ and $\{P_{t+1},P_{t+1}^\perp\}$. We are initially in
      an eigenstate $\ket{\psi_t}$ of the projector $P_t$, and want to
      transition to a state in the subspace $P_{t+1}$. With non-zero
      probability, the first $\{P_{t+1},P_{t+1}^\perp\}$ succeeds with outcome
      $P_{t+1}$. If it fails, we have prepared a state in the $P_{t+1}^\perp$
      subspace. We ``unwind'' the measurement by measuring $\{P_t,P_t^\perp\}$
      again. Upon repeating the $\{P_{t+1},P_{t+1}^\perp\}$ measurement, we
      again have a non-zero probability of successfully obtaining the
      $P_{t+1}$ outcome. If we fail again, we repeat the procedure until
      success.}
    \label{fig:measurement}
  \end{center}
\end{figure}

\begin{lemma}\label{pfailLabel}
  The measurement sequence depicted in \cref{fig:measurement} with the two
  projective measurements $\{P_t,P_t^\perp\}$ and $\{P_{t+1},P_{t+1}^\perp\}$
  has a failure probability bounded by
  \begin{equation}\label{failBnd}
    p_{fail}(m) < \frac{1}{2 \; d_{\min} m }
  \end{equation}
  after $m$-subsequent measurement steps, where $d_{\min} = \min_k d_k$ is the
  minimal overlap between the eigenstates of $P_t$ and $P_{t+1}$.
\end{lemma}
\begin{proof}
  Let $Q_1 = P_{t+1}$, $Q_0 = P_{t+1}^\perp$ and $R_1 = P_{t}$, $R_0 =
  P_{t}^\perp$, in accordance with the notation in \cref{Jordan}. Hence $Q_1$
  projects on to the new ground state subspace, whereas the $R_1$ is the
  projector on to the old ground state subspace. If we start in some state
  $\ket{\psi} = R_1\ket{\psi}$, the probability of failure of the measurement
  sequence depicted in \cref{fig:measurement} after $m$ steps is
  $p_{\mathrm{fail}}(m) = \sum_{s_1,\dots,s_m} \tr(Q_0 R_{s_m} Q_0 \dots
  R_{s_1} Q_0 \proj{\psi} Q_0 R_{s_1} \ldots Q_0 R_{s_{m}}
  Q_0)$.
  Note that $[Q_0R_sQ_0, Q_0R_pQ_0] = 0$ for all $s,p$. We can therefore
  rearrange this to express $p_{\mathrm{fail}}(m)$ as the sum
  \begin{multline*}
    \sum_{k=0}^m \binom{m}{k} \bra{\psi}
      \left (Q_0 R_{0} Q_0\right)^{2k}
      \left (Q_0 R_{1} Q_0\right)^{2(m-k)}
      \ket{\psi}\\
    = \bra{\psi} \left(
        \left(Q_0 R_0 Q_0\right)^2
        + \left(Q_0 R_1 Q_0\right)^2
      \right)^m \ket{\psi}.
  \end{multline*}
  If we work in the eigenbasis of $Q_1$, the individual $2\times2$
  block matrices take the form
  \begin{equation}
    Q_1^k = \begin{pmatrix} 1 & 0 \\ 0 & 0\end{pmatrix}, \quad
    R_1^k = \begin{pmatrix}
      d_k & \sqrt{d_k(1 - d_k)} \\
      \sqrt{d_k(1 - d_k)} & 1-d_k
    \end{pmatrix}.
  \end{equation}
  Since $\ket{\psi}$ is left invariant by $R_1$, we have that $\ket{\psi} =
  \sum_k c_k \ket{r_k}$, where in this basis every $\ket{r_k} = (\sqrt{d_k}
  \quad \sqrt{1-d_k})^T$ by \cref{Jordan}. We are therefore left with
  \begin{equation}
    p_{\mathrm{fail}}(m)
      = \sum_k \abs{c_k}^2\; (1- d_k)\left( 1 - 2d_k(1-d_k)\right)^m,
  \end{equation}
  with $d_k \in [0,1]$ and $\sum_k \abs{c_k}^2 =1$.

  Since $(1-x) \leq e^{-x}$, we may bound $(1-d_k)(1 - 2d_k(1-d_k))^m \leq
  (1-d_k)e^{- 2md_k(1-d_k)}$. Furthermore, we have that $(1-d_k) e^{-
    2md_k(1-d_k)} \leq 1/2md_k$ by Taylor expansion. If we now choose the
  largest factor $(2 m d_k)^{-1} \leq (2 m d_{\min})^{-1}$, we can bound the
  total failure probability by \cref{failBnd}.
\end{proof}

We use this to bound the overall runtime.
\begin{theorem}[Runtime]\label{runtime}
  Let $A^{v}$ be $G$-symmetric tensors defining a PEPS on an $N$-vertex
  lattice. A state in the subspace spanned by the corresponding
  G\nobreakdash-injective PEPS $\ket{A^1\dots A^N;K}$ can be prepared on a
  quantum computer with probability $1-\epsilon$ in time $\order{N^4
    \kappa_G^2 \Delta^{-1} \epsilon^{-1}}$, with additional classical
  processing $\order{N d^6}$, where $\Delta = \min_t(\Delta_t)$ is the minimal
  spectral gap of the family of parent Hamiltonians $H_t$ for $\ket{A^1\dots
    A^t}$ ($t=1\dots N$), and $\kappa_G = \max_t \kappa(A^t|_{S_G})$.
\end{theorem}
\begin{proof}
  The algorithm in \cref{alg:G-injective} first prepares the initial
  G\nobreakdash-isometric PEPS, which can be done exactly in time
  $\order{N\log N}$ \cite{Aguado}, and then transforms this step by step into
  the G\nobreakdash-injective PEPS, with one step for each of the $N$ vertices
  of $\mathcal{G}$. Each step has a probability of failure
  $p_{\mathrm{fail}}(m)$ if we repeat the back-and-forth measurement scheme
  $m$ times. We need to ensure that the total success probability is
  lower-bounded by $(1- p_{\mathrm{fail}}(m))^N \geq 1-\epsilon$. Since
  $(1-x)^N \geq 1 - Nx$, we can use \cref{pfailLabel} to bound
  \begin{equation}
    \left(1-p_{\mathrm{fail}}(m)\right)^N \geq 1 - \frac{N}{2m d_{\min}},
  \end{equation}
  so we want $N/2m d_{\min} \leq \epsilon$. Since $d_{\min} \geq
  \kappa_G^{-2}$ by \cref{dmin}, we choose $m \geq N\kappa_G^2/2\epsilon$ at
  each step. We therefore need to perform $\order{N^2
    \kappa_G^2\epsilon^{-1}}$ quantum phase estimation procedures, each of
  which has runtime $\tilde{O}(N^2/\Delta^{-1})$ to ensure that we are able to
  resolve the energy gap of the parent Hamiltonian \cite{HHL09}. (Note that
  the notation $\tilde{O}(\cdot)$ suppresses more slowly growing terms such as
  $\exp(\sqrt{\ln{N/\Delta}})$.) The classical bookkeeping required to keep
  track of the Hamiltonians is the same as in \cite{Martin}. Putting all this
  together, we arrive at the total runtime stated in the theorem.
\end{proof}

\paragraph{Discussion}
We have shown how the Marriott-Watrous rewinding technique combined with the
unique structure of G\nobreakdash-injective PEPS can be used to transition
from one state to the next (see \cref{alg:G-injective}), successively building
up the desired quantum state even when that state has topological order and
the ground states are degenerate. There are a number of alternative techniques
that could potentially be used to achieve the same thing. In each case, the
key to proving an efficient runtime is our \cref{dmin}. In many cases the
existing results in the literature assume non-degenerate ground states, so
would need to be generalised before they would apply to the topologically
degenerate ground states considered here.

Standard adiabatic state preparation has the disadvantage that it requires a
polynomial energy gap along a continuous path joining the initial Hamiltonian
with the final one. But the ``jagged adiabatic lemma'' of Ref.~\cite{ATS07}
shows that such a path connecting a discrete set of gapped Hamiltonians always
exists if the ground states are unique and each ground state has sufficient
overlap with the next. The latter is precisely what we prove in \cref{dmin}.
For the `injective' case of~\cite{Martin}, this is sufficient to show that
adiabatic state preparation is an efficient alternative to the
``rewinding trick''. Our results suggest that the jagged adiabatic lemma
could be generalised to the case of non-unique ground states.

More general are the methods of~\cite{BKS10}, which subsume the jagged
adiabatic lemma and the Marriott-Watrous technique. The results in~\cite{BKS10}
do not immediately apply to degenerate ground states, but if they can be
generalised they could potentially improve the polynomial dependence on the
required error probability to a logarithmic one.
Another potential alternative is the recent quantum rejection sampling
technique of~\cite{ORR12}, which gives quadratic improvement over
Marriott-Watrous rewinding by a clever use of amplitude amplification.
Finally, the spectral gap amplification technique of~\cite{SB11}, which cites
injective PEPS preparation~\cite{Martin} as a potential application, may also
be applicable. In all cases, the techniques would first need to be generalised
to handle degenerate ground states. If this can be done, our \cref{dmin} would
imply efficiency of the resulting algorithm.

The conditions required for efficient preparation in \cref{runtime}
(inverse-polynomial scaling with system size of the spectral gaps of the
partial parent Hamiltonians and polynomial scaling of the condition numbers of
the PEPS projectors) are very reminiscent of the conditions (local gap and
local topological quantum order) required for stability of the spectral gap of
local Hamiltonians~\cite{michalakis2012}. It is also conjectured that the
spectral gap of the parent Hamiltonian should be closely related to the
condition number of the PEPS projectors. It would be interesting to understand
better the relationships between these various conditions.

The technique we introduced, of constructing a complex many-body quantum state
by starting from an easily-constructable state and successively transforming
it into the desired state, is very general. Although we have applied it here
to G\nobreakdash-injective PEPS, as a class of states including many important
topological quantum states such as the RVB state, our algorithm can be
generalised to other classes of tensor network states, such as string-net
models~\cite{string-nets} and models constructed from Hopf
algebras~\cite{Oliver}.

\paragraph{Acknowledgements}
DPG and TSC would like to thank the hospitality of the Centro de Ciencias
Pedro Pascual and the Petronilla facility in Benasque, where part of this work
was carried out. TSC is funded by the Juan de la Cierva program of the Spanish
science ministry. DPG and TSC are supported by Spanish grants QUITEMAD,
I-MATH, and MTM2008-01366. MS is supported by Austrian SFB project FoQuS
F4014. KT is funded by the Erwin Schr\"odinger fellowship, Austrian Science
Fund (FWF): J~3219-N16. FV is supported by EU grants QUERG and by the Austrian
FWF SFB grants FoQuS and ViCoM.

\bibliography{G-injective}

\balancecolsandclearpage
\appendix

\section{Appendix: G-isometric PEPS}

In this appendix, we use the argument described in \cite{SCPG10} for the Toric
Code and RVB states, generalised here to arbitrary G\nobreakdash-isometric PEPS, to see
that the G\nobreakdash-isometric PEPS for any semi-regular representation is equivalent
to the one for the regular representation. Let us start with a semi-regular
representation $U_g$ of a group $G$, and let
\begin{equation}
  B = \frac{1}{\abs{G}}\sum_g \Delta\bar{U}_g\otimes \Delta\bar{U}_g
      \otimes \Delta U_g\otimes \Delta{U}_g\,.
\end{equation}
We will show how $B$ can indeed be seen as the G\nobreakdash-isometric PEPS
corresponding to the regular representation -- possibly composed with an
isometry which embeds the initial Hilbert space into a sufficiently large one.
As explained in the main text, the latter can be prepared efficiently on a
quantum computer by other means.

\begin{figure}[!htbp]
  \begin{center}
    \includegraphics[scale=1.2]{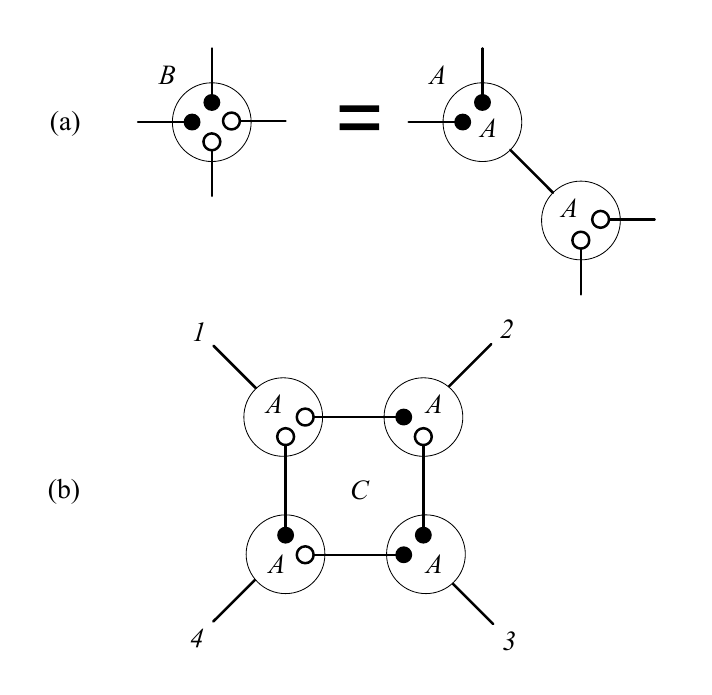}
    \caption{(a) illustrates the decomposition of the original tensor in the
      tensors $A$. We mark in white the bonds in which we have $U_g$ and in
      black those in which we have $\bar{U}_g$. (b) illustrates the new way of
      grouping the tensors to get a G-isometric PEPS, called C. The bonds of
      this new tensor are numbered clockwise as in the figure.}
    \label{fig:app}
  \end{center}
\end{figure}

To show this, we decompose the tensor $B$ into two tensors of the form
$A=(\sqrt{\abs{G}})^{-1}\sum_g \Delta U_g\otimes \Delta{U}_g\otimes\ket{g}$
(where $U_g$ and $U_g$ are interchanged as needed, as shown in
\cref{fig:app}(a)). By regrouping these new tensors, we obtain a new PEPS
decomposition of the same state, where now the bond dimension is $\abs{G}$
(\cref{fig:app}(b)). The resulting tensor $C$ (\cref{fig:app}(c)), as a map
from the virtual to the physical indices, is given by
\begin{equation}
  C:\ket{g_1g_2g_3g_4} \mapsto \frac{1}{\abs{G}^2}
    \Delta^2 U_{g_1g_2^{-1}}\ox\Delta^2 U_{g_2g_3^{-1}}
    \ox\Delta^2 U_{g_4g_3^{-1}}\ox\Delta^2 U_{g_1g_4^{-1}}.
\end{equation}
By calling $g=g_1^{-1}g'_1$ and using \cref{eq:delta} it is not difficult to
see that
\begin{multline}
  \BraKet{g'_1g'_2g'_3g'_4}{C^{\dagger}C}{g_1g_2g_3g_4}\\
  = \frac{1}{\abs{G}^4}\prod_{r=1}^2
    \tr(\Delta^4U_{g_r g_{r+1}^{-1} g'_{r+1} g_r^{'-1}})
    \prod_{r=3}^4 \tr(\Delta^4U_{g_{r+1} g_{r}^{-1} g'_{r} g_{r+1}^{'-1}})
\end{multline}
equals $1$ if and only if there exist $g$ such that $g_ig =g'_i$ for all $i$.
Otherwise, the expression is identically zero.

Therefore $C^\dagger C = (\abs{G})^{-1}\sum_g R_g^{\otimes 4}$ for the regular
representation $R_g$, hence the new PEPS $C$ is the G\nobreakdash-isometric
PEPS corresponding to the regular representation.

\end{document}